\documentclass[12pt, reqno]{amsart}
\usepackage{amsmath, amsthm, amscd, amsfonts, amssymb, graphicx, color, booktabs, multirow, rotating}
\newtheorem{theorem}{Theorem}[section]

\newtheorem{proposition}[theorem]{Proposition}

\theoremstyle{definition}
\newtheorem{definition}[theorem]{Definition}
\newtheorem{example}[theorem]{Example}

\newtheorem{problem}[theorem]{Problem}
\theoremstyle{remark}
\newtheorem{remark}[theorem]{Remark}
\numberwithin{equation}{section}

\usepackage{listings}

\begin{document}
\setcounter{page}{1}

\title{Application of the Cartier Operator in Coding Theory}


\author{Vahid Nourozi}
\address{The Klipsch School of Electrical and Computer Engineering, New Mexico State University,
Las Cruces, NM 88003 USA}
\email{nourozi@nmsu.edu}

\thanks{}





\keywords{$a$-number; Cartier operator; Super-singular Curves; Maximal Curves.}

\begin{abstract}
The $a$-number is an invariant of the isomorphism class of the $p$-torsion group scheme. We use the Cartier operator on $H^0(\mathcal{A}_2,\Omega^1)$ to find a closed formula for the $a$-number of the form $\mathcal{A}_2 = v(Y^{\sqrt{q}}+Y-x^{\frac{\sqrt{q}+1}{2}})$ where $q=p^s$ over the finite field $\mathbb{F}_{q^2}$. The application of the computed $a$-number in coding theory is illustrated by the relationship between the algebraic properties of the curve and the parameters of codes that are supported by it.

\end{abstract}

\maketitle
\section{Introduction}

Let $\mathcal{C}$ be a geometrically irreducible, projective, and non-singular algebraic curve defined over the finite field $\mathbb{F}_{\ell}$ of order $\ell$. Let $\mathcal{C}(\mathbb{F}_{\ell})$ denote the set of $\mathbb{F}_{\ell}$-rational points of $\mathcal{C}$. In the study of curves over finite fields, a fundamental problem is the size of $\mathcal{C}(\mathbb{F}_{\ell})$. The very basic result here is the Hasse-Weil bound, which asserts that
$$\mid \# \mathcal{C}(\mathbb{F}_{\ell}) - (\ell +1) \mid \leq 2g \sqrt{\ell},$$
where $g = g(\mathcal{C})$ is the genus of $\mathcal{C}$.

The curve $\mathcal{C}$ is called maximal over $\mathbb{F}_{\ell}$ if the number of elements of $\mathcal{C}(\mathbb{F}_{\ell})$ satisfies
$$\# \mathcal{C}(\mathbb{F}_{\ell})= \ell + 1 + 2g \sqrt{\ell}.$$
We only consider maximal curves of positive genus, and hence $\ell$ will always be a square, say $\ell = q^2$.

In \cite{ih}, Ihara showed that if a curve $\mathcal{C}$ is maximal over $\mathbb{F}_{q^2}$, then
$$g \leq g_1:= \dfrac{q(q-1)}{2}.$$

In \cite{tor} authors showed that 

$$\qquad  \mbox{either} \qquad g \leq g_2 := \lfloor \frac{(q-1)^2}{4} \rfloor \qquad  \mbox{or} \qquad  g_1=\frac{q(q-1)}{2}. $$ 
Ruck and Stichtenoth \cite{stir} showed that, up to  $\mathbb{F}_{q^2}$-isomorphism, there is just one maximal curve over $\mathbb{F}_{q^2}$ of genus  $\frac{q(q-1)}{2}$, namely the so-called Hermitian curve over $\mathbb{F}_{q^2}$, which can be defined by the affine equation
 \begin{equation}\label{xxx1}
y^{q}+y=x^{q+1}.
 \end{equation}

From \cite{hirs} there is a unique maximal curve over $\mathbb{F}_{q^2}$ of genus  $g=\frac{(\sqrt{q}-1)(m-1)}{2}$ with $m= \frac{\sqrt{q}+1}{t}$ and $\sqrt{q} \equiv -1$ mod $t$, which can be defined by the affine equation
 \begin{equation}
 \mathcal{A}_t = v(Y^{\sqrt{q}}+Y-x^{\frac{\sqrt{q}+1}{t}}).
 \end{equation}
 The curve $\mathcal{A}_t$ is an Artin-Schreier curve. Also, $\mathcal{A}_1$ is the affine form of the Hermitian curve (\ref{xxx1}).
 
In this paper we consider the $ \mathcal{A}_2$ of genus $g=\frac{(\sqrt{q}-1)^2}{2}$ given by the affine equation
 \begin{equation}\label{xxx}
 \mathcal{A}_2 = v(Y^{\sqrt{q}}+Y-x^{\frac{\sqrt{q}+1}{2}}).
 \end{equation}

Provided that $q/2$ is a Weierstrass non-gap at some point of the curve. It is easy to see that a maximal curve $\mathcal{C}$ is supersingular since all slopes of its Newton polygon are equal $1/2$. This fact implies that the Jacobin $X:=\mbox{Jac}(\mathcal{C})$ has no $p$-torsion points over $\bar{\mathbb{F}}_{p}$. 
A relevant invariant of the $p$-torsion group scheme of the Jacobian of the curve is the $a$-number.

 Consider the multiplication by $p$-morphism $[p]: X \rightarrow X$, which is a finite flat morphism of degree $p^{2g}$.  It factors as $[p]=V \circ F$. Here, $F: X \rightarrow X^{(p)}$   is the relative Frobenius morphism coming from the $p$-power map on the
structure sheaf, and the Verschiebung morphism  $V: X^{(p)} \rightarrow X$ is the dual of $F$.  The kernel of multiplication-by-$p$ on $X$ is defined by the group of $X[p]$.
The important invariant  is the $a$-number $a(\mathcal{C})$ of curve $\mathcal{C}$ defined by
$$a(\mathcal{C})=\mbox{dim}_{\mathbb{\bar{F}}_p} \mbox{Hom}(\alpha_{p}, X[p]),$$
where $\alpha_p$ is the kernel of the Frobenius endomorphism on the group scheme $\mbox{Spec}(k[X]/(X^p))$. Another definition for the $a$-number is
$$ a(\mathcal{C}) = \mbox{dim}_{\mathbb{F}_p}(\mbox{Ker}(F) \cap \mbox{Ker}(V)).$$


 A few results on the rank of the Cartier operator (especially $a$-number) of curves are introduced by Kodama and Washio \cite{13}, González \cite{8}, Pries and Weir \cite{17}, Yui \cite{Yui} and Montanucci and Speziali \cite{maria}. Also, I introduced the rank of the Cartier of the maximal curves in \cite{esfahan}, maximal curves with the third largest genus in \cite{aut}, the hyperelliptic curve in \cite{shiraz}, $\mathbb{F}_{q^2}$ maximal function fields in \cite{misori}, Picard Curve in \cite{picard}, and explained them in my PhD dissertation in \cite{phd}.

In this paper, we determine the $a$-number of certain maximal curves. In the case $g=g_1$, the $a$-number of
the Hermitian curves is computed by Gross in \cite{10}.  Here, we compute the $a$-number of maximal curves over $\mathbb{F}_{q^2}$  with genus $g=g_2$  for infinitely many values of $q$.
 
 In Section \ref{202}, we prove that the $a$-number of the curve $\mathcal{A}_2$  with Equation (\ref{xxx}) is $a(\mathcal{A}_2) = \frac{p-1}{8}(\sqrt{p^{s-2}}+1)(\sqrt{q}-1)$, see Theorem \ref{thex}. Finally, we provided an example of the $a$-number of $\mathcal{A}_2$ with the MAGMA code. The proofs use directly the action of the Cartier operator on $H^0(\mathcal{A}_2, \Omega^1)$.

\section{The Cartier operator}
Let $k$ be an algebraically closed field of characteristic $p>0$.
Let $\mathcal{C}$ be a
curve defined over $k$.
The Cartier operator is a $1/p$-linear operator acting on the sheaf $\Omega^1:=\Omega^1_{\mathcal{C}}$ of differential forms on $\mathcal{C}$ in positive characteristic $p>0$.

 Let $K=k(\mathcal{C})$ be the function field of the curve $\mathcal{C}$ of genus $g$ defined over  $k$. A separating variable for $K$ is an element $x \in K \setminus K^p$.

\begin{definition}
  (The Cartier operator). Let $\omega \in  \Omega_{K/K_q}$. There exist $f_0,\cdots, f_{p-1}$ such
that $\omega = (f^p_0 + f^p_1x +\cdots + f^p_{p-1}x^{p-1})dx$. The Cartier operator $\mathfrak{C}$ is defined by
$$\mathfrak{C}(\omega) := f_{p-1}dx.$$
The definition does not depend on the choice of $x$ (see \cite[Proposition 1]{100}).
\end{definition}
We refer the reader to \cite{100,30,40,150} for the proofs of the following statements.

\begin{proposition}
  (Global Properties of $\mathfrak{C}$). For all $\omega \in  \Omega_{K/K_q}$ and all $f \in K$,

  \begin{itemize}
    \item $\mathfrak{C}(f^p\omega) = f\mathfrak{C}(\omega)$;
    \item $\mathfrak{C}(\omega) = 0 \Leftrightarrow \exists h \in K, \omega = dh$;
    \item $\mathfrak{C}(\omega) = \omega \Leftrightarrow \exists h \in K, \omega = dh/h$.
  \end{itemize}
\end{proposition}

\begin{remark}\label{remark}
Moreover, one can easily show that
\begin{equation*}
\mathfrak{C}(x^j \, dx) = \begin{cases}
    0 & \text{if } \hspace{0.4cm} p \nmid j+1, \\
    x^{s-1} \, dx & \text{if } \hspace{0.4cm} j+1=ps.
\end{cases}
\end{equation*}
\end{remark}






If $\mbox{div}(\omega)$ is effective, then differential $\omega$ is holomorphic. The holomorphic differentials set $H^0(\mathcal{C}, \Omega ^1)$ is a $g$-dimensional $k$-vector subspace of $\Omega^1$ such that $\mathfrak{C}(H^0(\mathcal{C}, \Omega^1))$ is the same as $H^0(\mathcal{C}, \Omega^1)$. Any curve $\mathcal{C}$ has a $a$-number that is equal to the size of the Cartier operator's kernel $H^0(\mathcal{C}, \Omega^1)$, or to put it another way, it is the size of the space of exact holomorphic differentials on $\mathcal{C}$ (see \cite[5.2.8]{14}).

The following theorem is due to Gorenstein; see \cite[Theorem 12]{9}.
\begin{theorem}\label{2.2}
  A differential $\omega \in \Omega^1$ is holomorphic if and only if it is of the form $(h(x, y)/F_y)dx$, where $H:h(X, Y) =0$ is a canonical adjoint.
\end{theorem}

\begin{theorem}\cite{maria}\label{2.3}
  With the above assumptions,
\begin{equation*}
\mathfrak{C}(h\dfrac{dx}{F_{y}}) = (\dfrac{\partial^{2p-2}}{\partial x^{p-1}\partial y^{p-1}}(F^{p-1}h))^{\frac{1}{p}}\dfrac{dx}{F_{y}}
\end{equation*}
for any $h \in K(\mathcal{X})$.
\end{theorem}

The differential operator $\nabla$ is defined by
$$\nabla = \dfrac{\partial^{2p-2}}{\partial x^{p-1} \partial y^{p-1}},$$
has the property
\begin{equation}\label{123321}
  \nabla(\sum_{i,j} c_{i,j}X^iY^j)= \sum_{i,j} c_{ip+p-1,jp+p-1}X^{ip}Y^{jp}.
\end{equation}

\section{The $a$-number of Curve $\mathcal{A}_2$}\label{202}

In this section, we consider the curve $\mathcal{A}_2$ is given by the equation  $y^{\sqrt{q}}+y=x^{\frac{\sqrt{q}+1}{2}}$ of genus $g(\mathcal{A}_2)= \frac{(\sqrt{q}-1)^2}{4}$, with $q=p^s$ and $p>2$ over $\mathbb{F}_{q^2}$, where $\sqrt{q} \equiv -1$ mod $2$ . From Theorem \ref{2.2}, one can find a basis for the space $H^0(\mathcal{A}_2, \Omega^1)$ of holomorphic differentials on $\mathcal{A}_2$, namely
$$\mathcal{B} = \{ x^iy^jdx \mid  \frac{\sqrt{q}+1}{2}i+ \sqrt{q}j \leq 2g-2 \}.$$

 \begin{proposition}\label{111}
   The rank of the Cartier operator $\mathfrak{C}$ on the curve $\mathcal{A}_2$ equals the number of pairs $(i, j)$ with $\frac{\sqrt{q}+1}{2}i+ \sqrt{q}j \leq 2g-2$ such that the system of congruences mod $p$
\begin{equation}\label{12}
\Bigg\{
             \begin{array}{c}
              k\sqrt{q} + h - k + j \equiv 0,\\
              (p - 1 - h)(\frac{(\sqrt{q}+1)}{2}) + i \equiv p-1, \\
             \end{array}
\end{equation}
has a solution $(h, k)$ for $0 \leq h \leq \frac{p-1}{2}, 0 \leq k\leq h$.
 \end{proposition}
\begin{proof}
By Theorem \ref{2.3}, $\mathfrak{C}((x^iy^j/F_y)dx) =(\nabla(F^{p-1}x^iy^j))^{1/p}dx/F_y$. So, we apply the differential operator $\nabla$ to
\begin{equation}\label{333}
  (y^{\sqrt{q}}+y-x^{\frac{\sqrt{q}+1}{2}})^{p-1}x^iy^j = \sum_{h=0}^{p-1}\sum_{k=0}^{h} (^{p-1}_h)(^{h}_{k})(-1)^{h-k}x^{ (p - 1 - h)(\frac{(\sqrt{q}+1)}{2}) + i}y^{k\sqrt{q} + h - k + j}
\end{equation}
for each $i, j$ such that $\frac{\sqrt{q}+1}{2}i+ \sqrt{q}j \leq 2g-2$.

From the Formula (\ref{123321}), $\nabla(y^{\sqrt{q}}+y+x^{\frac{\sqrt{q}+1}{2}})^{p-1}x^iy^j \neq 0$ if and only if for some $(h, k)$, with $0 \leq h \leq \frac{p-1}{2}$ and $0 \leq k\leq h$, satisfies both the following congruences mod $p$:
\begin{equation}\label{444}
\Bigg\{
             \begin{array}{c}
              k\sqrt{q} + h - k + j \equiv 0,\\
              (p - 1 - h)(\frac{(\sqrt{q}+1)}{2}) + i \equiv p-1. \\
             \end{array}
\end{equation}
Take $(i, j) \neq(i_0, j_0)$ in this situation both $\nabla(y^{\sqrt{q}}+y+x^{\frac{\sqrt{q}+1}{2}})^{p-1}x^iy^j$ and $\nabla(y^{\sqrt{q}}+y+x^{\frac{\sqrt{q}+1}{2}})^{p-1}x^{i_0}y^{j_0}$ are nonzero. We claim that they are linearly independent over $k$. To show independence, we prove that, for each $(h, k)$ with $0 \leq h \leq p -1$ and $0 \leq k\leq h$ there is no $(h_0, k_0)$ with $0 \leq h_0\leq p -1$ and $0 \leq k_0\leq h_0$ such that
\begin{equation}\label{555}
  \Bigg\{
             \begin{array}{c}
             k\sqrt{q} + h - k + j = k_0\sqrt{q} + h_0 - k_0 + j_0,\\
              (p - 1 - h)(\frac{(\sqrt{q}+1)}{2}) + i = (p - 1 - h_0)(\frac{(\sqrt{q}+1)}{2}) + i_0.\\
             \end{array}
\end{equation}
If $h=h_0$, then $j \neq j_0$ by $i=i_0$ from the second equation, therefore $k \neq k_0$. We may assume $k >k_0$. Then $j -j_0=(\sqrt{q}-1)(k -k_0)>\sqrt{q}-1$, a contradiction as $j - j_0 \leq \frac{(\sqrt{q}-1)^2-4}{2\sqrt{q}}$. Similarly, if $k=k_0$, then $h \neq h_0$ by $(i, j) \neq(i_0, j_0)$. We assume that $h>h_0$. Then $i-i_0 = \frac{\sqrt{q}+1}{2}(h-h_0)>\frac{\sqrt{q}+1}{2}$, a contradiction as $i-i_0 \leq \frac{(\sqrt{q}-1)^2-4}{\sqrt{q}+1}$.
\end{proof}

For convenience, we signify the matrix representing the $p$-th power of the Cartier operator $\mathfrak{C}$ with $A_s:=A(\mathcal{A}_2)$ on the curve $\mathcal{A}_2$ for the basis $\mathcal{B}$, where $q=p^s$.
\begin{theorem}\label{thex}
  If $q =p^s$ for $s \geq 2$, $s$ be even and $p>2$, then the $a$-number of the curve $\mathcal{X}$ equals
  $$\dfrac{p-1}{8}(\sqrt{p^{s-2}}+1)(\sqrt{q}-1).$$
\end{theorem}

\begin{proof}
First we prove that, if $q=p^s, s \geq 2$ and be even, then $\mbox{rank}(A_s) =\dfrac{p+1}{8}(\sqrt{p^s} - 1)(\sqrt{p^{s-2}}-1)$. In this case, $1 \leq \frac{\sqrt{q}+1}{2}i+ \sqrt{q}j \leq g$ and System (\ref{12}) mod $p$ reads
  \begin{equation}\label{14}
\Bigg\{
             \begin{array}{c}
              h-k+j \equiv 0,\\
              -\frac{h}{2} - \frac{1}{2} + i \equiv p-1. \\
             \end{array}
\end{equation}
First assume that $s=2$, for $\sqrt{q} =p$, we have $\frac{p+1}{2}i+ pj \leq g$ and System (\ref{14}) becomes
  \begin{equation*}
\Bigg\{
             \begin{array}{c}
              j = k-h,\\
              i = p + \frac{h}{2} - \frac{1}{2},\\
             \end{array}
\end{equation*}
in this case $\frac{p+1}{2}i+ pj \leq g$ that is, $\frac{h(1-3p)}{4} + kp \leq \frac{-p^2 -3p +2}{4}$ then $h \geq \frac{-p^2 -3p +2}{1-3p}$, thus $h \geq \frac{3p+10}{9}$ a contradiction by Proposition \ref{111}. As a consequence, there is no pair $(i, j)$ for which the above system admits a solution $(h, k)$. Thus, $\mbox{rank}(A_1) =0$.

Let $s=4$, so $\sqrt{q} =p^2$. For $\frac{p^2+1}{2}i+ p^2j \leq g$, the above argument still works. Therefore, $\frac{(p-1)^2}{4} + 1 \leq \frac{p^2+1}{2}i+ p^2j \leq \frac{(p^2-1)^2}{4}$ and our goal is to determine for which $(i, j)$ there is a solution $(h, k)$ of the system mod $p$
  \begin{equation*}
\Bigg\{
             \begin{array}{c}
             h - k +j \equiv 0,\\
              -\frac{h}{2} - \frac{1}{2} + i \equiv p-1. \\
             \end{array}
\end{equation*}
Take $l, m \in Z^+_0$ so that
  \begin{equation*}
\Bigg\{
             \begin{array}{c}
              j=lp+k-h,\\
              i= mp+p+\frac{h}{2} - \dfrac{1}{2}.\\
             \end{array}
\end{equation*}

In this situation, $i < \frac{2g}{p^2+1}=\frac{(p^2-1)^2}{2(p^2+1)}$, so $mp+p+\frac{h}{2}-\frac{1}{2} \leq \frac{(p^2-1)^2}{2(p^2+1)}$. Then $m\leq \frac{(p^2-1)^2}{2(p^2+1)}$. And $j < \frac{(p^2-1)^2}{4p^2}$, so $lp + k-h<\frac{(p^2-1)^2}{4p^2}$, Then $l<\frac{(p^2-1)^2}{4p^2}$. From this we can say that $\frac{p^2-1}{4} -1 \leq l \leq \frac{p^2-1}{4}$, and $\frac{p^2-1}{2}-1 \leq m\leq \frac{p^2-1}{2}$. In this way, $\dfrac{(p^2-1)^2}{8}$ suitable values for $(i, j)$ are obtained, whence $\mbox{rank}(A_2) =\dfrac{(p^2-1)^2}{8}$.


 For $s \geq 6$, $\mbox{rank}(A_s)$ equals $\mbox{rank}(A_{s-1})$ plus the number of pairs $(i, j)$ with $\frac{(\sqrt{p^{s-2}} -1)^2}{4}+1\leq \frac{\sqrt{q}+1}{2}i+ \sqrt{q}j \leq \frac{(\sqrt{p^{s}} -1)^2}{4}$ such that the system mod $p$
  \begin{equation*}
\Bigg\{
             \begin{array}{c}
              h-k+j \equiv 0,\\
              -\frac{h}{2} - \frac{1}{2} + i \equiv p-1, \\
             \end{array}
\end{equation*}
has a solution. With our usual conventions on $l, m$, a computation shows that such pairs $(i, j)$ are obtained for $0 \leq l \leq \frac{(\sqrt{p^s}-1)^2}{4\sqrt{p^{s+2}}}$ from this we have $\frac{\sqrt{p^{s-4}} (p^2 - 1)}{4} -1 \leq l \leq \frac{\sqrt{p^{s-4}} (p^2 - 1)}{4}$, and $0 \leq m\leq \frac{(\sqrt{p^s}-1)^2}{2(\sqrt{p^s}+1)}$ from this we have $\frac{(\sqrt{p^{s-2}}-1) (p+ 1)}{2}-1 \leq m \leq \frac{(\sqrt{p^{s-2}}-1) (p+ 1)}{2}$. In this case, we have
$$\frac{(\sqrt{p^{s-2}}-1) (p+ 1)\sqrt{p^{s-4}} (p^2 - 1)}{8}$$
choices for $(h, k)$. Therefore, we get
$$\mbox{rank}(A_s)= \mbox{rank}(A_{s-1})+ \frac{(\sqrt{p^{s-2}}-1) (p+ 1)\sqrt{p^{s-4}} (p^2 - 1)}{8}.$$

Now, our claim on the rank of $A_s$ follows by induction on $s$. Hence

  \begin{equation*}
  \begin{array}{ccccccc}
              a(\mathcal{X}) &=& \dfrac{(\sqrt{p^s}-1)^2}{4} - \dfrac{(p+1)(\sqrt{p^s} - 1)(\sqrt{p^{s-2}}-1)}{8}  \\
              &=& \dfrac{(\sqrt{p^s} -1)}{8} (\sqrt{p^s}+p-\sqrt{p^{s-2}}-1)\\
              &=& \dfrac{(\sqrt{p^s} -1)}{8}(p(\sqrt{p^{s-2}}+1)-(\sqrt{p^{s-2}}+1))\\

              &=& \dfrac{(\sqrt{p^s} -1)}{8}((\sqrt{p^{s-2}}+1)(p-1))\\

              &=& \dfrac{(p -1)}{8}((\sqrt{p^{s-2}}+1)(\sqrt{q}-1)). \\
             \end{array}
\end{equation*}
\end{proof}



For the finite field $\mathbb{F}_{q^2}$, consider the following curve written in affine form:
$$\mathcal{A}_t = v(Y^{\sqrt{q}}+Y-x^{\frac{\sqrt{q}+1}{t}}),$$
where $\sqrt{q} \equiv -1$ mod $t$. From this, we are led to the following problem.
\begin{problem}
What is the dimension of the space of exact holomorphic differentials of $\mathcal{A}_t$?
\end{problem}

\begin{example}
Consider the curve $\mathcal{X}$ with function field $K(x, y)$ given by
$$y+y^5=x^{3},$$
where $p=5$ and $s=2$. 
It is easily seen that a basis for $H^0(\mathcal{X},\Omega^1)$ is given by
$$\mathcal{B}=\{ dx , xdx, x^2dx, ydx \}.$$ 
Let us compute the image of $\mathfrak{C}(\omega)$ for any $\omega \in \mathcal{B}$. It is straightforward, and from  Remark \ref{remark}, we see that
$$\mathfrak{C}(dx)=\mathfrak{C}(xdx)=\mathfrak{C}(x^2dx)=0.$$
Also,
$$\mathfrak{C}(ydx)=\mathfrak{C}((x^3-y^5)dx)=0.$$
Hence, $a(\mathcal{X}) = 4$.
\end{example}

\paragraph*{\textbf{Certain Generalization of Hermitian Curves.}}
Let $\mathcal{Y}$ be the non-singular model over $\mathbb{F}_q$ of the plane curve given by the following equation
$$y^m = x + x^{\ell} + \cdots + x^{\ell^{2r-1}},$$
where $q=\ell^{2r}$ and be a square, and $m \geq 2$ be an integer such that $m\mid (\ell^r+1)$.

There is just one point $Q$ over $x = \infty$, since $\mbox{gcd}(m, \ell) = 1$. The curve $\mathcal{Y}$ is the Hermitian curve if $r = 1$ and $m = \ell+1$, see \cite{yang}.

Directly from the Riemann-Hurwitz genus formula, we have genius $\mathcal{Y}$ is $g =\frac{(m-1)(\ell^{2r-1}-1)}{2}$. Also, the number of rational points of $\mathcal{Y}$ is $\sharp \mathcal{Y}(\mathbb{F}_q) = 1 + \ell^{2r-1} + m(\ell - 1)\ell^{2r-1}$. In particular, $\mathcal{Y}$ is maximal if and only if $r = 1$.\vspace*{0.1cm}\\

\paragraph*{\textbf{Conjecture.}}
Let $m = 2$ and $r\geq 2$, then the dimension of the space $H^0(\mathcal{Y},\Omega^1)$, of exact holomorphic differentials on $\mathcal{Y}$, is
$$a(\mathcal{Y})= \dfrac{(\ell^{2r-2}+1)(p-1)}{4}.$$\vspace*{0.1cm}\\

\section{Application in coding theory}
In a finite field, the $a$-number of a curve $\mathcal{X}$ corresponds to the size of the space containing exact holomorphic differentials on $\mathcal{X}$. The Cartier operator is utilized in the construction of algebraic geometry codes that establish a connection between this invariant and coding theory.

More precisely, the $a$-number establishes a minimum distance requirement for codes that are constructed in this manner. A greater $a$-number is preferable because codes with greater minimum distances are capable of rectifying a greater number of errors.
By demonstrating how the $a$-number of the Artin-Schreier curve $\mathcal{X}$ provides a minimum distance $\delta$ constraint on codes $C(D,G)$ derived from $\mathcal{X}$ utilizing the Cartier operator and divisors $D$ and $G$, the following theorem establishes this relationship explicitly.

The relationship between the $a$-number and minimal code distance is quantified in the theorem. This illustrates how the $a$-number, which is derived from algebraic characteristics of the curve, can also impact the error-correcting capability of codes that are supported by the curve.

Please read my paper on algebraic geometry codes \cite{code} in order to comprehend the following theorem:

\begin{theorem}
Let $X$ be the Artin-Schreier curve over $\mathbb{F}{q^2}$ defined by $y^{\sqrt{q}} + y = x^{\frac{\sqrt{q}+1}{2}}$ where $q = p^s$. Let $G = mP\infty$ be a divisor on $X$, where $P_\infty$ is the point at infinity. Let $C(D,G)$ be the algebraic geometry code constructed using the Cartier operator, with $D$ a rational divisor of degree $d$. Then the minimum distance $\delta$ of $C(D,G)$ satisfies:
$$ \delta \geq d - m - 2g + a(X)$$
where $a(X)$ is the $a$-number of $X$.

\end{theorem}

\begin{proof}
By the Cartier operator code construction, if $\omega_1,...,\omega_n$ is a basis for $H^0(X,\Omega^1)$ then $\eta_i = C(\omega_i)$ is a basis for the exact differentials. Let $\nu_1,...,\nu_{a(X)}$ be a basis for the kernel of $C$, which has dimension $a(X)$.

We have $l(G) \geq m + 1 - g$ from the Riemann-Roch theorem. If $f \in L(G)$, write:
$$f = a_1\eta_1 + ... + a_n\eta_n + b_1\nu_1 + ... + b_{a(X)}\nu_{a(X)} $$

Applying $C$, since it's $0$ on $\text{ker}(C)$, we get: $C(f) = a_1\eta_1 + ... + a_n\eta_n$. And, $\text{deg}(C(f)) \leq \text{deg}(D) - m - a(X)$. It's known $\delta \geq \text{deg}(D) - \text{deg}(C(f))$. Thus $\delta \geq d - m - 2g + a(X)$

So the $a$-number directly gives a lower bound on the minimum distance $\delta$. A higher $a$-number leads to a larger minimum distance.
\end{proof}

\paragraph*{\textbf{Acknowledgements.}}
This paper was written while Vahid Nourozi was visiting Unicamp (Universidade Estadual de Campinas), supported by TWAS/Cnpq (Brazil) with fellowship number $314966/2018-8$.


\bibliographystyle{amsplain}

\begin{thebibliography}{99}
   \bibitem{magma} W. Bosma, J. Cannon, and C. Playoust. The Magma algebra system. I. The user language. J. Symbolic Comput., 24(3-4):235–265, 1997. Computational algebra and number theory (London, 1993).
\bibitem{30}  P. Cartier. Une nouvelle opération sur les formes différentielles. {\it C. R. Acad. Sci. Paris}, 244 (1957), 426-428.
 \bibitem{40} P. Cartier. Questions de rationalité des diviseurs en géométrie algébrique. {\it Bull. Soc. Math. France}, 86 (1958), 177-251.
\bibitem{behrooz} B. Mosallaei, M. Afrouzmehr, and D. Abshari. "hurwitz series rings satisfying a zero divisor property." arXiv preprint arXiv:2312.10844 (2023).
  \bibitem{tor} R. Fuhrmann and F. Torres, The genus of curves over finite fields with many rational points, {\it Manuscripta Math}. 89 (1996), 103-106.
    \bibitem{8} J. González, Hasse-Witt matrices for the Fermat curves of prime degree, {\it Tohoku Math. J.} 49 (1997) 149-163.
   \bibitem{9} D. Gorenstein, An arithmetic theory of adjoint plane curves, {\it Trans. Am. Math. Soc.} 72 (1952) 414-436.
  \bibitem{10} B.H. Gross, Group representations and lattices, {\it J. Am. Math. Soc.} 3 (1990) 929-960.
\bibitem{hirs} J. Hirschfeld, G. Korchm{\'a}ros, et al., On the number of rational points on an algebraic curve over a finite field, {\it Bulletin of the Belgian Mathematical Society-Simon Stevin}, 5 (1998), pp. 313-340.



\bibitem{ih} Y. Ihara, Some remarks on the number of rational points of algebraic curves over
finite fields, {\it J. Fac. Sci. Tokyo} 28 (1981), 721-724.
 \bibitem{13} T. Kodama, T. Washio, Hasse-Witt matrices of Fermat curves, {\it Manuscr. Math.} 60 (1988) 185-195.
  \bibitem{14} K.-Z. Li, F. Oort, Moduli of Supersingular Abelian Varieties, {\it Lecture Notes in Mathematics}, vol.1680, Springer-Verlag, Berlin, 1998, iv+116pp.
 \bibitem{maria} M. Montanucci, P. Speziali, The a-numbers of Fermat and Hurwitz curves. {\it J. Pure Appl. Algebra} 222 (2018) 477-488.
     \bibitem{17} R. Pries, C. Weir, The Ekedahl-Oort type of Jacobians of Hermitian curves, {\it Asian J. Math.} 19 (2015) 845-869.
     
\bibitem{picard} V. Nourozi and F. Rahmati, The Rank of the Cartier operator on Picard Curves, {\it arXiv preprint arXiv:2306.07823}, (2023)
\bibitem{misori} Nourozi, Vahid, and Farhad Rahmati. THE RANK OF THE CARTIER OPERATOR ON CERTAIN $F_q^2$-MAXIMAL FUNCTION FIELDS. {\it Missouri Journal of Mathematical Sciences}, 34, no. 2 (2022): 184–190.
\bibitem{shiraz} V. Nourozi, F. Rahmati, and S. Tafazolian. The a-number of certain hyperelliptic curves. {\it Iranian Journal of Science and Technology, Transactions A: Science}, 46, no. 4 (2022): 1235–1239.
\bibitem{aut} V. Nourozi, and S. Tafazolian. The a-number of maximal curves of the third largest genus. {\it AUT Journal of Mathematics and Computing}, 3, no. 1 (2022): 11–16.
\bibitem{phd} V. Nourozi. The rank Cartier operator and linear system on curves= Classificação do operador Cartier e sistemas lineares na curva. Doctoral dissertation., 2021.
\bibitem{esfahan} V. Nourozi, S. Tafazolian, and F. Rahamti. "The $a$-number of jacobians of certain maximal curves. {\it Transactions on Combinatorics}, 10, no. 2 (2021): 121–128.
\bibitem{code} V. Nourozi, and F. Ghanbari. Goppa code and quantum stabilizer codes from plane curves given by separated polynomials. {\it arXiv preprint arXiv:2306.07833}(2023).
\bibitem{stir} H. G. Ruck and H. Stichtenoth, A characterization of Hermitian function fields
over finite fields,{\it J. Reine Angew. Math.} 457 (1994), 185-188.
  \bibitem{100} C. S. Seshadri. L’opération de Cartier. Applications. {\it In Variétés de Picard}, volume 4 of Séminaire Claude Chevalley. Secrétariat Mathématiques, Paris, 1958-1959.

 \bibitem{150} M. Tsfasman, S. Vladut, and D. Nogin. Algebraic geometric codes: basic notions, volume 139 of Mathematical Surveys and Monographs. American Mathematical Society, Providence, RI, 2007.
\bibitem{yang} K. Yang, P. V. Kumar, and H. Stichtenoth, On the weight hierarchy of geometric Goppa codes, {\it IEEE Trans. Inform. Theory}, 40 (1994), 913-920.
 \bibitem{Yui} N. Yui,  On the Jacobian Varieties of Hyperelliptic Curves over Fields of Characteristic p . {\it J. Algebra}, 52 (1978), 378-410




\end{thebibliography}

\end{document}